\documentclass[letterpaper, 10 pt, conference]{ieeeconf}  

\IEEEoverridecommandlockouts                              
\usepackage{amsmath,amsthm,amssymb,amsfonts,mathrsfs}
\usepackage{algorithmic,algorithm}
\usepackage{leftidx}
\usepackage{array}
\usepackage{textcomp}
\usepackage{stfloats}
\usepackage{url}
\usepackage{verbatim}
\usepackage{graphicx}
\usepackage{cite}
\usepackage{bm}

\usepackage{hyperref}

\usepackage{soul, color, xcolor}
\soulregister{\cite}7 
\soulregister{\citep}7 
\soulregister{\citet}7 
\soulregister{\ref}7 
\soulregister{\pageref}7 

\theoremstyle{definition}
\newtheorem{theorem}{\textbf{Theorem}}

\newtheorem{proposition}{\textbf{Proposition}}
\newtheorem{lemma}{\textbf{Lemma}}
\newtheorem{remark}{\textbf{Remark}}

\newtheorem{definition}{\textbf{Definition}}

\usepackage{scalerel,stackengine}
\stackMath
\newcommand\reallywidehat[1]{%
	\savestack{\tmpbox}{\stretchto{%
			\scaleto{%
				\scalerel*[\widthof{\ensuremath{#1}}]{\kern.1pt\mathchar"0362\kern.1pt}%
				{\rule{0ex}{\textheight}}
			}{\textheight}%
		}{2.4ex}}%
	\stackon[-6.9pt]{#1}{\tmpbox}%
}

\usepackage{threeparttable}

\title{Enhanced Robust Tracking Control: An Online Learning Approach}

\author{Ao Jin, Weijian Zhao, Yifeng Ma, Panfeng Huang, and Fan Zhang$^*$ 
\thanks{Ao Jin, Weijian Zhao, Yifeng Ma, Panfeng Huang, and Fan Zhang are with Research Center for Intelligent Robotics, Shaanxi Province Innovation Team of Intelligent Robotic Technology, School of Astronautics, Northwestern Polytechnical University, Xi’an 710072, China (E-mail: jinao@mail.nwpu.edu.cn, fzhang@nwpu.edu.cn)}%
}

\begin{document}

\maketitle
\thispagestyle{empty}
\pagestyle{empty}

\begin{abstract}
This work focuses the tracking control problem for nonlinear systems subjected to unknown external disturbances. Inspired by contraction theory, a neural network-dirven CCM synthesis is adopted to obtain a feedback controller that could track any feasible trajectory. Based on the observation that the system states under continuous control input inherently contain embedded information about unknown external disturbances, we propose an online learning scheme that captures the disturbances dyanmics from online historical data and embeds the compensation within the CCM controller. The proposed scheme operates as a plug-and-play module that intrinsically enhances the tracking performance of CCM synthesis. The numerical simulations on tethered space robot and PVTOL demonstrate the effectiveness of proposed scheme. The source code of the proposed online learning scheme can be found at {\url{https://github.com/NPU-RCIR/Online_CCM.git}}.
\end{abstract}

\section{Introduction}
In robotic applications, autonomous systems often necessitate interactions with unstructured environments containing unknown or time-varying dynamics, while being susceptible to unknown external disturbances such as unmodeled forces, sensor noise, and payload variations. For the tracking control problem, these unknown external disturbances can significantly degrade the performance of the closed-loop system. Moreover, the presence of environmental uncertainties typically leads a conservative motion planning to ensure safety constraint satisfaction. This fundamental trade-off between robustness and performance underscores the critical need for control architectures capable of online disturbance estimation and compensation. 

The control contraction metric (CCM) synthesis provides a systematic framework of designing the tracking controller for nonlinear systems with performance guarantees \cite{Manchester2017}. The CCM synthesis is based on the contraction theory \cite{Lohmiller1998,Manchester2014}, which studies the exponential convergence of any two neighboring trajectories of nonlinear system. The CCM synthesis can be formulated as a convex optimization problem, which can be solved by sum-of-squares (SOS) programming \cite{Wei2021,Singh2023}. Recently, many learning-based schemes have been proposed for finding a valid CCM and the associated tracking controller \cite{Sun2021,Tsukamoto2021,Rezazadeh2022,Richards2023,Song2024}, which provides a new perspective for the CCM synthesis.

\textbf{\textit{Main challenges and motivations}}: The CCM synthesis admits a feedback controller that could track any feasible trajectory with guaranteed performance. The SOS programming has been successfully applied for searching the valid metric. However, the applicability of this methodology to general robotic systems remains constrained by the prerequisite that the dynamics needs to be polynomial equations or can be approximated by polynomial equations. In addition, for the nonlinear control-affine systems subjected to external disturbances, the tracking performance exhibits a direct correlation with the boundness of external disturbances. A robust CCM (RCCM) \cite{Zhao2022} has been proposed to minimize the $\mathcal{L}_{\infty}$ gain from disturbances to the tracking error. But the controller proposed in \cite{Zhao2022} needs to solve a nonlinear programming (NLP) problem online, which is computationally expensive. The $\mathcal{L}_1$-adaptive control is cooperated with CCM synthesis \cite{Lakshmanan2020} to improve the performance with state and time-varing uncertainties. However, its tracking performance relys on the design of the adaptive controller. The reinforcement learning (RL) \cite{Wang2023i} is introduced to the CCM synthesis for estimating the external disturbances, which has expensive computation cost for estimating disturbances.

\begin{figure}[!t]
	\centering
	\includegraphics[width=16pc]{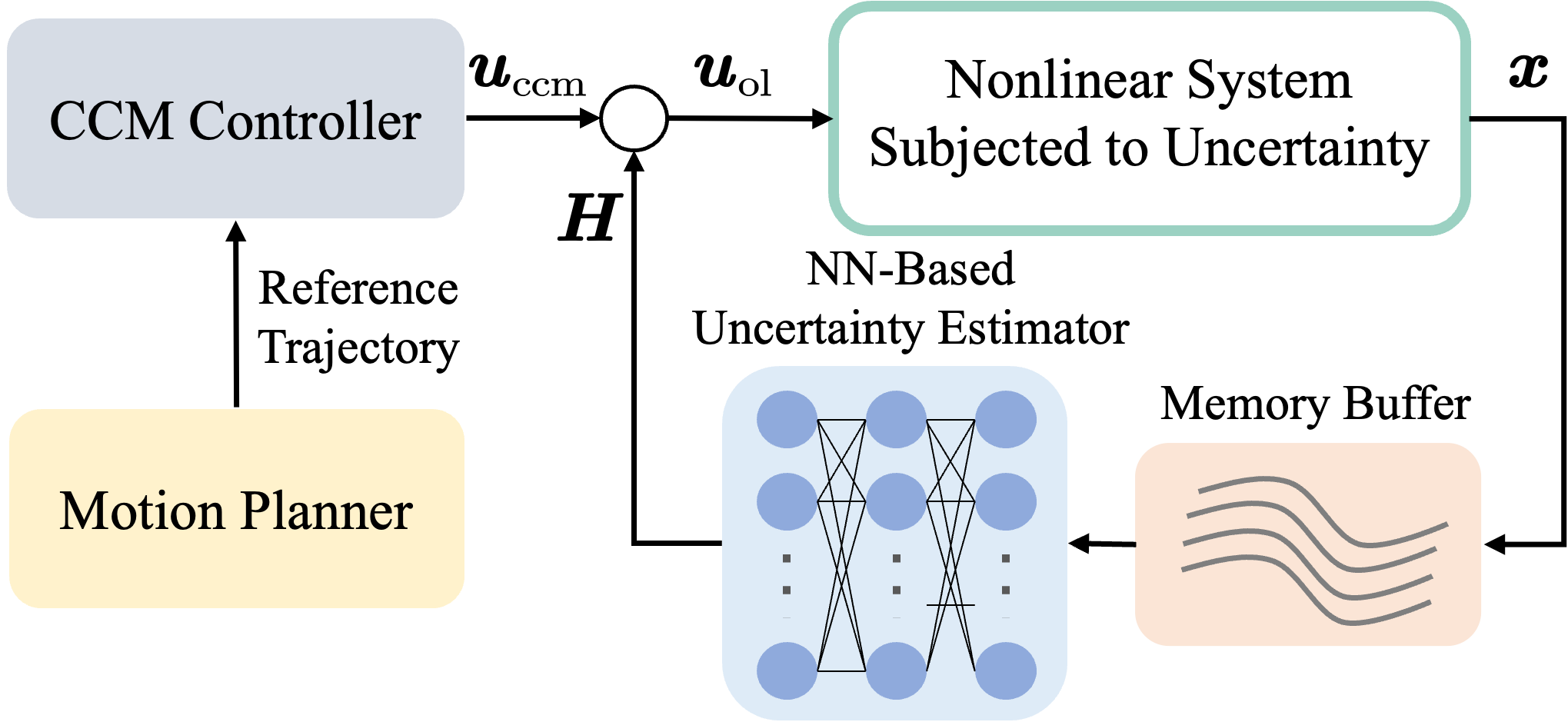} 
	\caption{Illustration of proposed online learning scheme for enhancing the robust tracking control.} \label{framework}
\end{figure}

Therefore, in this work, we propose an online learning scheme to enhance the trakcing performance of the CCM synthesis. First, the valid metric and tracking controller are jointly learned by the neural network-driven CCM synthesis \cite{Sun2021}, which eliminates the online computation of the geodesic. Then, by utilizing the online historical data, an online learning scheme is proposed for capturing the unknown external disturbances. Specifically, a virtual dynamics which is augmented from the nominal dynamics, is established for the perturbed system. A neural network is integrated into the virtual dynamics to estimate the external disturbances. The parameters of the neural network are updated online by minimizing the discrepancy between the trajectories of virtual dynamics and perturbed system in a moving horizon manner. The estimated disturbances are compensated into the tracking controller directly, which can improve the tracking performance of the closed-loop system and refine the "tube" derived from CCM synthesis. The proposed online learning scheme operates as a plug-and-play module that intrinsically enhances the performance of CCM synthesis without necessitating supplementary controller. The proposed framework is illustrated in Fig \ref{framework}.

The main contributions of this work are summarized as follows: 1) An online learning scheme is proposed to capture the unknown external disturbances from online historical data. 2) The compensation of external disturbances enhances the tracking performance of the neural network-driven CCM synthesis. 3) The "tube" derived from CCM synthesis can be refined online for a reduced conservatism motion planning.

\textbf{\textit{Notation}}: Given a matrix $\bm M$, define $\hat{\bm M}=\bm M+\bm M^\mathrm{T}$. The $\mathcal{T}_x\mathcal{X}$ represents the tangent space of $\mathcal{X}$ at $\bm x\in \mathcal{X}$. The maximum and minimum singular values of matrix $\bm M$ are represented by $\overline{\sigma}(\bm M)$ and $\underline{\sigma}(\bm M)$, respectively. The $\Vert \bm x \Vert_M=\sqrt{x^\mathrm{T}\bm  My}$ denotes the weighted form for matrix $\bm M\in S_n^{>0}$.

\section{Problem Formulation}
In this work, we consider a nonlinear system subjected to unknown external disturbances in a control-affine form
\begin{equation}
	\dot{\bm{x}} = \bm f(\bm x) + \bm B(\bm x)(\bm u + \bm h(t, \bm x)) \label{eq_system}
\end{equation}
where $\bm{x} \in \mathcal{X} \subset \mathbb{R}^n$ is the state, $\bm u \in\mathcal{U} \subset \mathbb R^m$ is the control input, $\mathcal{X}$ and $\mathcal{U}$ denote the state and control constraint set. The functions $\bm f:\ \mathbb{R}^n \rightarrow \mathbb{R}^n$ and $\bm{B}:\ \mathbb{R}^n\rightarrow\mathbb{R}^{n\times m}$ are known and Lipschitz continuous. The function $\bm h(t, \bm x)$ represents the unknown external disturbances with the known upper bound $\overline{H}$, and it is Lipschitz continuous. The unperturbed (nominal) dynamics of system (\ref{eq_system}) can be formulated as
\begin{equation}
    \dot{\bm{x}} = \bm{f}(\bm{x})+\bm{B}(\bm{x})\bm{u} \label{eq_nominal_system}
\end{equation}

Consider a motion plan $(\bm{x}^*(t),\bm{u}^*(t))$ generated from any planner based on the unperturbed dynamics (\ref{eq_nominal_system})
\begin{equation}
    \dot{\bm{x}}^*(t) = \bm{f}(\bm{x}^*(t)) + \bm{B}(\bm{x}^*(t))\bm{u}^*(t)
\end{equation}
The objective under consideration is to design a control law $\bm{u}$ for nonlinear system (\ref{eq_system}) to track the motion plan $(\bm{x}^*(t),\bm{u}^*(t))$ as closely as possible, while keeping the state $\bm{x}(t)$ in a compact safe set $\mathcal{X}_{\text{safe}}\subset \mathcal{X}$, for all $t>0$.

\section{Preliminaries on Control Contraction Metrics}
The basic idea of control contraction metric (CCM) synthesis is to analyze the incremental exponential stability (IES) of a system by investigating the evolution of infinitesimal distance between any two neighboring trajectories. 

For the unperturbed system (\ref{eq_nominal_system}), we consider a pair of neighboring trajectories and denote the infinitesimal displacement between them by $\bm{\delta_x}$. The dynamics of this infinitesimal displacement is given by 
\begin{equation}
	\dot{\bm{\delta}_x}=\bm A(\bm{x},\bm{u})\bm{\delta}_x+ \bm{B}(\bm{x})\bm{\delta}_u
\end{equation}
where $\bm A=\frac{\partial \bm{f}}{\partial \bm{x}}+\sum_{i=1}^{m}\frac{\partial \bm b_i}{\partial \bm x}\bm u_i$, $\bm b_i$ is the $i$ th column of $\bm{B}(\bm{x})$ and $\bm u_i$ is $i$ th component of $\bm{u}$. Define the Riemann squared length of $\bm{\delta_x}$ by \cite{Manchester2017}
\begin{equation}
    \bm{V}(\bm{x},\bm{\delta_x}) = {\bm{\delta}_x^\mathrm{T}}\bm{M}(\bm{\bm{x}})\bm{\delta_x}
\end{equation}
where $\bm{M}(\bm{x})$ is a smooth function which maps from $\mathcal{X}$ to the set of uniformly positive definite symmetric matrices. As shown in \cite{Singh2023}, the CCM is defined as follows
\begin{definition}\label{CCM}
	If there exists a uniformly bounded metric $\bm M(\bm{x})$, and a differential controller $ \bm{\delta}_u: \mathcal{T}_x\mathcal{X}\rightarrow\mathcal{T}_u\mathcal{U}$,  such that $\dot{\bm{V}}(\bm{x}, \bm{\delta}_x)\prec 0, \forall  \bm{x}\in\mathcal{X}, \forall \bm{\delta}_x\in\mathcal{T}_x\mathcal{X} $. Then, the metric $\bm M(\bm{x})$ is referred as a CCM of the system. Additionally, if $\dot{\bm{V}}(\bm{x}, \bm{\delta}_x)\prec -2\lambda \bm{V}(\bm{x}, \bm{\delta}_x), \forall  \bm{x}\in\mathcal{X}$ holds for a positive constant $\lambda$, then the system is contracting at rate $\lambda$.
\end{definition}

If a feedback controller $\bm{u}=\bm{k}(\bm{x})+\bm{v}$ that makes the closed-loop system contracting at rate $\lambda$ and a given metric $\bm{M}(\bm{x})$ for any continuous signal $\bm{v}$, then the following inequality holds \cite{Lohmiller1998}
\begin{equation}
	\dot{\bm{M}} + \reallywidehat{\bm{M}\left(\bm A+\bm B\bm K \right)}+2\lambda \bm M \prec 0 \label{condition_ccm}
\end{equation}
where $\dot{\bm{M}}=\sum_{i=1}^{n}\frac{\partial \bm M}{\partial x_i}\dot{\bm{x}},\bm K=\frac{\partial \bm k}{\partial \bm{x}}$. The inequality (\ref{condition_ccm}) gives a sufficient condition for searching a valid CCM and a feedback controller that makes the system contracting at rate $\lambda$. Moreover, when the feedback controller is designed as $\bm{u} = \bm{k}(\bm{x},\bm{x}^*)+\bm{u}^*$ satisfying $\bm{u}=\bm{u}^*$ when $\bm{x}=\bm{x}^*$, the unperturbed closed-loop system owns the following property. 
\begin{proposition} \label{proposition_unperturbed_track}
	If the condition (\ref{condition_ccm}) is satisfied with a uniformly bounded metric $\bm M(\bm{x})$, i.e., $\underline{\alpha}\bm{I}\preceq \bm M(\bm{x})\preceq \overline{\alpha}\bm{I} $ and a feedback tracking controller $\bm{u}=\bm k(\bm{x},\bm{x}^*)+\bm{u}^*$, then the displacement between actual trajectory and nominal trajectory of unperturbed close-loop system exponentially converges to zero, i.e.,
	\begin{equation}
		\Vert  \bm{x}(t)- \bm{x}^*(t)\Vert \le R e^{-\lambda t}\Vert  \bm{x}(0)- \bm{x}^*(0)\Vert \label{converge}
	\end{equation}
	where $\lambda$ is the contracting rate and the constant $R=\sqrt{{\overline{\alpha}}/{\underline{\alpha}}}$ is the overshoot constant. 
\end{proposition}

If the closed-loop system is subjected to the external disturbances, the CCM synthesis also gives the performance guarantees.
\begin{lemma} \cite{Singh2023} \label{lemma_rci}
	If there is a uniformly bounded metric $\bm M(\bm{x})$, i.e., $\underline{\alpha}\bm{I}\preceq \bm M(\bm{x})\preceq \overline{\alpha}\bm{I} $ and a feedback controller $\bm{u}=\bm k(\bm{x},\bm{x}^*)+ \bm{u}^*$ satisfying the condition (\ref{condition_ccm}), then, for the same reference trajectory, the distance between actual trajectory of perturbed closed-loop dynamics and nominal trajectory of unperturbed dynamics is bounded. And 
	\begin{equation}
		\Omega(\bm{x}^*):= \left\{ \bm{x}\in \mathcal{X}:\left\| \bm{x}- \bm{x}^* \right\| _{\underline{M}}^2\le \overline{c}^2 \right\} \label{RCI_set}
	\end{equation}
	is an ellipsoid robust control invariant (RCI) set of closed-loop system stabilized with $\bm u=\bm k(\bm x, \bm x^*)+\bm u^*$ , where $\overline{c}=\sup_{ \bm{x}\in\mathcal{X}}\overline{\sigma}(\Theta(\bm{x})\bm{B}(\bm{x}))\overline{H}/\lambda, \bm M(\bm x)=\Theta(\bm{x})^T\Theta(\bm{x})$, $\bm M(\bm x)\ge \underline{M}$ holds for all $\bm x\in \mathcal{X}$, and $\overline{H}$ is the upper bound of external disturbances. 
\end{lemma}

\section{Main Results}
\subsection{Neural Network-Driven CCM Synthesis}
As shown in Lemma (\ref{lemma_rci}), the CCM synthesis can be extended to the tracking control problem with performance guarantees. In this work, we propose a neural network-driven CCM synthesis to search for the valid metric and feedback controller. The metric is parameterized by a neural network $\bm{M}(\bm{x};\theta_{\bm{M}})$ with parameters $\theta_{\bm{M}}$, and the feedback gain is also parameterized by a neural network $\bm k(\bm x,\bm x^*;\theta_{\bm{k}})$ with parameters $\theta_{\bm{k}}$. The controller $\bm u(\bm x,\bm x^*,\bm u^*;\theta_{\bm{k}})$ is given by $\bm u(\bm x,\bm x^*,\bm u^*;\theta_{\bm{k}}) = \bm u^*+\bm k(\bm x,\bm x^*;\theta_{\bm{k}})$, which satisfies that $\bm{u}(\bm{x},\bm{x}^*,\bm{u}^*;\theta_{\bm{k}})=\bm{u}^*$ when $\bm{x}=\bm{x}^*$ for all $\theta_{\bm{k}}$. The risk function for searching the valid metric and feedback controller is defined as
\begin{equation}
	L_{\text{ccm}} = \mathop {\mathbb{E}} \limits_{\left( \bm x,\bm x^*,\bm u^* \right) \sim \rho \left( \mathcal{D} \right)} \mathcal{G}(-\psi(\bm x,\bm x^*,\bm u^*;\theta_{\bm{M}},\theta_{\bm{u}})-\sigma) \label{loss_ccm}
\end{equation}
where $\psi(\bm x,\bm x^*,\bm u^*;\theta_{\bm{M}},\theta_{\bm{u}})$ represents the LHS of (\ref{condition_ccm}), $\rho(\mathcal{D})$ denotes the uniform distribution over the sampling data space $\mathcal{D}:=\mathcal{X}\times\mathcal{X}\times\mathcal{U}$, $\sigma$ is a small positive constant,  and the auxiliary function $ \mathcal{G}$ is constrained for positive semi-definite, i.e., $ \mathcal{G}(\bm A)=0$ if and only if $\bm A\succeq 0$. 

However, the searching space for finding the valid metric and feedback controller by minimizing (\ref{loss_ccm}) is very large, which would lead poor performance of the neural network-driven CCM synthesis. To address this issue, we impose a tighter condition for searching the valid metric and feedback controller. A stronger condition for searching valid metric and feedback controller is presented in \cite{Manchester2017}, which can be expressed as
\begin{gather}
	\bm B_{\bot}^\mathrm{T}[-\partial_{\bm{f}}\bm W+\reallywidehat{\frac{\partial \bm{f}}{\partial \bm{x}}\bm W}+2\lambda \bm W ]\bm B_{\bot} \prec 0 \label{c1} \\
	\bm B_{\bot}^\mathrm{T}[\partial_{\bm b_i}\bm W-\reallywidehat{\frac{\partial \bm b_i}{\partial \bm{x}}\bm W}]\bm B_{\bot} = 0 ,i=1,...,m \label{c2}
\end{gather}
where $\bm{W}(\bm{x}) = \bm{M}(\bm{x})^{-1}$ is the dual metric, and $B_{\bot}$ satisfies $B_{\bot}^\mathrm{T} B=0$. By imposing constraints (\ref{c1}) and (\ref{c2}), the risk functions are defiend by
\begin{gather}
	L_{C_1} = \mathop {\mathbb{E}} \limits_{\left( \bm x,\bm x^*,\bm u^* \right) \sim \rho \left( \mathcal{D} \right)}\mathcal{G}(-\bm C_1(\bm x;\theta_M)) \\
	L_{C_2} = \mathop {\mathbb{E}} \limits_{\left(\bm x,\bm x^*,\bm u^* \right) \sim \rho \left( \mathcal{D} \right)} \Vert \bm{C}_2\Vert_F
\end{gather}
where $\bm C_1(\bm x;\theta_M)$ represents the LHS of (\ref{c1}), $\bm{C}_2:=(C_2^1,...,C_2^i,...,C_2^m)$, $C_2^i$ denotes the LHS of (\ref{c2}). 

In addition, it should impose constraint to guarantee that the metric is uniformly bounded. Construct the metric as $\bm M(\bm x;\theta_{\bm{M}}) = \underline{\alpha}\bm{I}+ \bm m(\bm x;\theta_{\bm{M}})^\mathrm{T} \bm m(\bm x;\theta_{\bm{M}})$ where $\bm{m}(\bm{x};\theta_{\bm{M}})$ is a neural network with parameters $\theta_{\bm{M}}$. This representation ensures that the smallest condition number of metric is lower bounded by $\underline{\alpha}$ for all $\bm{x}\in\mathcal{X}$. The largest condition number risk function is raised by
\begin{equation}
	L_{\bm M} = \mathop \mathbb{E} \limits_{\left( \bm x,\bm x^*,\bm u^* \right) \sim \rho \left( \mathcal{D} \right)}\mathcal{G}(\overline{\alpha}\bm{I}-\bm M(\bm x;\theta_M))
\end{equation}

The empirical risk function for searching the valid metric and feedback controller is given by
\begin{equation}
	\begin{aligned}
		\mathcal{L}(\theta_M,\theta_u,\lambda,\overline{\alpha},\underline{\alpha})&=\frac{1}{N}\sum_{i=1}^{N}[L_{\text{ccm}}(\bm x_i,\bm x_i^*,\bm u_i^*)+L_{C_1}(\bm x_i)\\
		&+L_{C_2}(\bm x_i)+L_M(\bm{x}_i)] \label{loss_function}
	\end{aligned} 
\end{equation}
where the tuple $(\bm x_i,\bm x_i^*,\bm u_i^*)_{i=1}^N$ contains $N$ samples distributed uniformly over the data space $\mathcal{D}$.

With the learned metric and feedback controller, one can derive the tightened state and control constraints for the motion planner
\begin{equation}
	\begin{aligned}
		&\bm x^*\in \overline{\mathcal{X}}=\mathcal{X}\ominus\Omega, \bm u^*\in \overline{\mathcal{U}} =\{\bm{\overline{u}}\in\mathcal{U}|\overline{\bm u}+\bm k(\bm x,\bm x^*)\in\mathcal{U},\\ &\bm x\in \Omega(\bm x^*),\forall \bm x \in \mathcal{X},\forall \bm x^* \in \mathcal{\overline{X}}\}
	\end{aligned}
\end{equation}
where $\Omega(\bm x^*)$ is the RCI set defined in (\ref{RCI_set}).

\begin{remark}
	As shown in \cite{Manchester2017,Singh2021}, if $\bm B(\bm x)$ has the sparse representation as $\bm{B}(\bm{x})=[\bm{O}_{(n-m)\times m}\ \bm{b}(\bm{x})]^{\mathrm{T}}$, where the matrix $\bm b(\bm x)\in\mathbb{R}^{m\times m}$ is invertible, the condition (\ref{c2}) would be satisfied automatically if the upper-left $(n-m)\times(n-m)$ block of the dual metric $\bm W(\bm x)$ is not a function of the last $m$ components of $\bm{x}$.
\end{remark}


\subsection{Enhanced Robust Tracking Control via Online Learning}
As shown in Lemma \ref{lemma_rci}, once the search for the valid metric and related feedback controller is done, the tracking performance of perturbed closed-loop system is directly correlated with the upper bounds of external disturbances. Generally, the upper bounds of external disturbances are often conservatively given, resulting in conservative trajectory tracking performance of the CCM synthesis. Theoretically, the system states under continuous control input inherently contain embedded information about unknown external disturbances, which can also be derived from (\ref{eq_system}). This observation can be exploited for capturing the unknown external disturbances in a data-driven manner. Therefore, we propose an online learning approach to capture the unknown external disturbances from online historical data for enhancing the performance of CCM synthesis. The core concept is stated as follows.

For the perturbed system (\ref{eq_system}), consider a virtual dynamics 
\begin{equation}
    \dot{\bm{\xi}} = \bm{f}(\bm{\xi}) + \bm{B}(\bm{\xi})(\bm{u} + \bm{H}(t,\bm{\xi}) )\label{eq_virtual_system}
\end{equation}
where $\bm{\xi}\in \mathbb{R}^n$ is the virtual state, $\bm{u}$ is the control input of (\ref{eq_system}), $\bm{H}(t,\bm{\xi})$ denotes the expected external disturbances. Given an initial state $\bm{\xi}(t_0)$ and a control law $\bm{u}(t)$, one can get a virtual trajectory
\begin{equation}
    \bm{\xi}(t) = \bm{\xi}(t_0) + \int_{t_0}^{t}(\bm{f}(\bm{\xi}) + \bm{B}(\bm{\xi})(\bm{u} + \bm{H}(\tau,\bm{\xi}) ))d\tau \label{virtual_trajectory}
\end{equation}
And for the perturbed system (\ref{eq_system}), with the same control law $\bm{u}(t)$, the actual trajectory can be obtained by
\begin{equation}
    \bm{x}(t) = \bm{x}(t_0) + \int_{t_0}^{t}(\bm{f}(\bm{x}) + \bm{B}(\bm{x})(\bm{u} + \bm{h}(\tau, \bm{x}) ))d\tau \label{actual_trajectory}
\end{equation}
where $\bm{x}(t_0)$ is the initial state of (\ref{eq_system}), and $\bm{h}(\tau, \bm{x})$ is the actual external disturbances. 

\begin{theorem}
	If $\bm{\xi}(t_0)=\bm{x}(t_0)$ and the virtual trajectory $\bm{\xi}(t)$ is identical to the actual trajectory $\bm{x}(t)$ for any $[t_0,t]$, then the virtual dynamics is equivalent to the actual system and $\bm{H}(t,\bm{\xi})$ can be regarded as the true observation of unknown external disturbances.
\end{theorem}
\begin{proof}
	When the virtual trajectory equals to the actual trajectory with the same initial condition, one can derive 
	\begin{equation}
		\int_{t_0}^{t}\bm{p}(\tau,\bm{x})d \tau=0
	\end{equation}
	holds for any $[t_0,t]$, and $\bm{p}(\tau,\bm{x}) = \bm{H}(\tau,\bm{x})-\bm{h}(\tau,\bm{x})$. Assume that there exist ${\tau}_1$ and $\bm{x}_1$ such that $\bm{p}({\tau}_1,\bm{x}_1)>0$. And since $\bm{p}(\tau,\bm{x})$ is Lipschitz continuous, there exists a neighborhood $U = [{\tau}_1-\delta,{\tau}_1+\delta]\times [\bm{x}_1-\epsilon,\bm{x}_1+\epsilon]$ such that $\bm{p}(\tau,\bm{x})>0$ holds. Define $\bm{x}(\tau)=\bm{x}_0$ when $\tau\in[\tau_1-\delta,\tau_1+\delta]$, and $\bm{x}=0$ when $\tau \notin [{\tau}_1-\delta,{\tau}_1+\delta]$. One has $\int_{{\tau}_1-\delta}^{{\tau}_1+\delta}\bm{p}(\tau,\bm{x})d\tau >0$, which contradicts the hypothesis. Therefore, $\bm{p}({\tau},\bm{x})=0$ holds for ${\tau}>0$. This completes the proof. 
\end{proof}

Therefore, one can obtain the estimation of unknown external disturbances by solving the following optimization problem for any $[t_0,t]$
\begin{equation}
	\begin{aligned}
		\mathop{\min\limits_{\bm{H}(\tau,\bm{\xi})}} \ \  & \Vert \bm{\xi}(\tau) - \bm{x}(\tau) \Vert \\
		\mbox{s.t.}\quad  & \bm{\xi}(t_0) = \bm{x}(t_0),\ t_0 \le \tau \le t 
	\end{aligned} \label{opt_prob}
\end{equation}
However, since the external disturbances $\bm{h}(t,\bm{x})$ are unknown, with only the initial state $\bm{x}_0$ and control law $\bm{u}$ provided, the actual trajectory $\bm{x}(t)$ cannot be computationally determined by (\ref{actual_trajectory}). In addition, the optimization problem (\ref{opt_prob}) is infinite dimensional. To mitigate these issues, the online historical data is utilized to replace the actual trajectory $\bm{x}(t)$, and the optimization problem is solved in a moving horizon manner. 

\begin{figure}[!t]
	\centering
	\includegraphics[width=18pc]{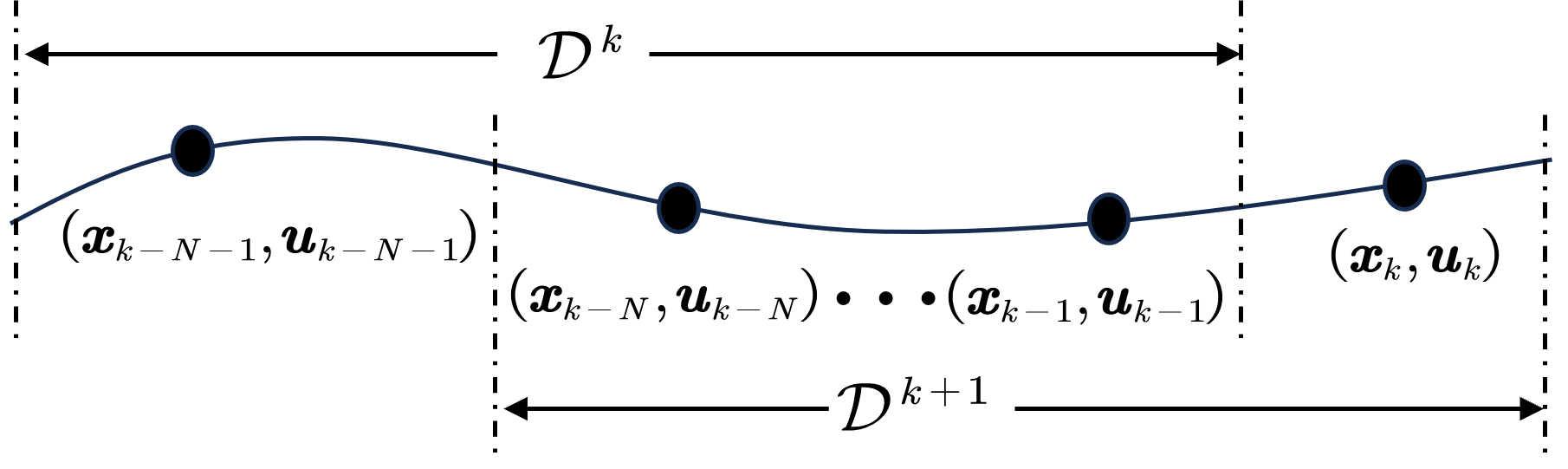} 
	\caption{The schematic of construction of memory buffer.} \label{memory_buffer_construction}
\end{figure}

Consider a memory buffer at time instant $k$
\begin{equation}
    \mathcal{D}^k=\{(\bar{\bm{x}}_i^k, \bar{\bm{u}}_i^k)\}_{i=1}^{N} \label{memory_buffer}
\end{equation}
where $\bar{\bm{x}}_i^k$, $\bar{\bm{u}}_i^k$ denote the historical state and control input of perturbed system (\ref{eq_system}) at time instant $k-1-N+i$, respectively. The memory buffer has a capacity of $N$ tuples. The construction of memory buffer is illustrated in Fig. \ref{memory_buffer_construction}. And a multilayer perceptron network is adopted to capture the external disturbances. The neural network can be expressed by
\begin{equation}
	\bm H(\bm{y};\theta_{\bm{H}}) = \bm{W}_2\phi (\bm{W}_1\phi (\bm{W}_0\bm{y}+\bm{b}_0) +\bm{b}_1) + \bm{b}_2 \label{nn_disturbance}
\end{equation}
where $\bm{W}_0\in \mathbb{R}^{{(n+1)}\times h}$, $\bm{w}_1\in \mathbb{R}^{h\times h}$, $\bm{W}_2\in \mathbb{R}^{h\times m}$ are the weight matrices, $\bm{b}_0\in \mathbb{R}^{h}$, $\bm{b}_1\in \mathbb{R}^{h}$, $\bm{b}_2\in \mathbb{R}^{m}$ are the bias vectors, $\phi$ is the activation function, and $h$ is the number of neurons in the hidden layers, $\theta_{\bm{H}}$ denotes the parameters to be optimized, $\bm{y}\triangleq[t,\bm{x}]$ represents the input of the neural network.

Similar to (\ref{memory_buffer}), a virtual memory buffer is also constructed for the virtual dynamics (\ref{eq_virtual_system}) at time instant $k$ by (\ref{virtual_trajectory})
\begin{equation}
    \mathcal{D}^k_{\text{virtual}}=\{(\bar{\bm{\xi}}_i^k, \bar{\bm{u}}_i^k)\}_{i=1}^{N} \label{memory_buffer_virtual}
\end{equation}
where $\bar{\bm{\xi}}_i^k$ and $\bar{\bm{u}}_i^k$ the historical state and control input of virtual dynamics (\ref{eq_virtual_system}) at time instant $k-1-N+i$, respectively. Now, at time instant $k$, the optimization problem (\ref{opt_prob}) can reformulated as
\begin{equation}
    \min_{\theta_{\bm{H}}} = \Vert \bar{\bm{\xi}}_k^N - \bar{\bm{x}}_k^N \Vert \label{opt_prob_reduced}
\end{equation} 
where $\bar{\bm{\xi}}_k^N$ is the virtual trajectory constructed by $\mathcal{D}^k_{\text{virtual}}$ at time instant $k$, and $\bar{\bm{x}}_k^N$ is the actual trajectory constructed by $\mathcal{D}^k$ at time instant $k$. At time instant $k+1$, the memory buffer (\ref{memory_buffer}) and (\ref{memory_buffer_virtual}) are reconstructed with new measured data, and the parameters $\theta_{\bm{H}}$ are updated by solving the optimization problem (\ref{opt_prob_reduced}) with the new memory buffer. 

With the estimated external disturbances $\bm{H}(t,\bm{x};\theta_{\bm{H}})$, the control law $\bm{u}$ can be redesigned as
\begin{equation}
    \bm{u}_{\text{ol}} =  \bm{u}_{\text{ccm}}- \bm{H}(t,\bm{x};\theta_{\bm{H}}) \label{control_law_ol}
\end{equation}
where $\bm{u}_{\text{ccm}}$ is the control law designed by the neural network-driven CCM synthesis. The proposed enhanced robust tracking control scheme is summarized in Algorithm \ref{algorithm_ol}. The tracking performance of the perturbed closed-loop system under the control law (\ref{control_law_ol}) is analyzed as follows.
\begin{theorem} \label{theorem_rci_refined}
    For the perturbed system (\ref{eq_system}), if there is a uniformly bounded metric $\bm M(\bm{x})$, i.e., $\underline{\alpha}\bm{I}\preceq \bm M(\bm{x})\preceq \overline{\alpha}\bm{I} $ and a feedback controller $\bm{u}_{\text{ccm}}=\bm k(\bm{x},\bm{x}^*)+ \bm{u}^*$ satisfying the condition (\ref{condition_ccm}). Then, under the control law (\ref{control_law_ol}), the distance between actual trajectory of perturbed closed-loop dynamics and nominal trajectory of unperturbed dynamics is bounded. And the RCI set $\Omega_{\text{ol}}(\bm{x}^*):= \left\{ \bm{x}\in \mathcal{X}:\left\| \bm{x}- \bm{x}^* \right\| _{\underline{M}}^2\le \overline{c}^2 \right\}$ is an ellipsoid RCI set of closed-loop system stabilized with $\bm{u}_{\text{ol}} =  \bm{u}_{\text{ccm}}- \bm{H}(t,\bm{x};\theta_{\bm{H}})$ , where $\overline{c}=\sup_{ \bm{x}\in\mathcal{X}}\overline{\sigma}(\Theta(\bm{x}))\overline{e}_H/\lambda, \bm M(\bm x)=\Theta(\bm{x})^T\Theta(\bm{x})$, $\bm M(\bm x)\ge \underline{M}$ holds for all $\bm x\in \mathcal{X}$, and $\overline{e}_H$ is the upper bound of external disturbances estimation error.
\end{theorem}
\begin{proof}
    The perturbed closed-loop system under the control law (\ref{control_law_ol}) can be formulated as
    \begin{equation}
        \dot{\bm{x}} = \bm{f}(\bm{x}) + \bm{B}(\bm{x})(\bm{u}_{\text{ccm}}+\bm{e}_H) \label{eq_eh_dynamics}
    \end{equation}
    where $\bm{e}_H = \bm{h}(t,\bm{x}) - \bm{H}(t,\bm{x};\theta_{\bm{H}})$ represents the estimation error of unknown external disturbances. Therefore, the RCI set can be derived by Lemma \ref{lemma_rci}. The proof is completed.
\end{proof}

Now, the tracking performance of perturbed closed-loop system under the control law (\ref{control_law_ol}) is related to the estimation error of external disturbances rather than the upper bounds of them. In addition, by compensating the estimation of unknown external disturbances into the feedback controller $\bm{u}_{\text{ccm}}$, the tube size of RCI set can be refined as shown in Theorem \ref{theorem_rci_refined}. The tightened state and control constraints are refined as
\begin{equation}
	\begin{aligned}
		&\bm x^*\in \overline{\mathcal{X}}_{\text{ol}}=\mathcal{X}\ominus\Omega_{\text{ol}}, \bm u^*\in \overline{\mathcal{U}}_{\text{ol}} =\{\bm{\overline{u}}\in\mathcal{U}|\overline{\bm u}+\bm k(\bm x,\bm x^*)\in\mathcal{U},\\ &\bm x\in \Omega_{\text{ol}}(\bm x^*),\forall \bm x \in \mathcal{X},\forall \bm x^* \in \mathcal{\overline{X}}_{\text{ol}}\}
	\end{aligned}
\end{equation}
With the refined state and control constraints, one can re-generate the reference trajectory $(\bm{x}^*,\bm{u}^*)$ online for a motion planning with mitigated conservatism. 


\begin{figure}[!t]
	\renewcommand{\algorithmicensure}{\textbf{Tracking Control}}
	\begin{algorithm}[H] 
		\caption{Enhanced Robust Tracking Control vis Online Learning}
		\label{algorithm_ol} 
		\begin{algorithmic}[1]
			\STATE Input: Tracking controller $\bm{u}_{\text{ccm}}$, nominal reference trajectory $(\bm x^*,\bm u^*)$.
			memory buffer $\mathcal{D}^k$ and $\mathcal{D}^k_{\text{virtual}}$.
			\STATE Initialization: $k=0$.
			\WHILE {The tracking task is not completed}
			\STATE Get the new measurement at time instant $k$.
			\IF{The memory buffers $\mathcal{D}^k$ and $\mathcal{D}^k_{\text{virtual}}$ are not full.}
			\STATE Track the nominal trajectory $(\bm x^*(t),\bm u^*(t))$ with the controller $\bm u_{\text{ccm}}$.
			\ELSE 
			\STATE Update the memory buffers $\mathcal{D}^k$ and $\mathcal{D}^k_{\text{virtual}}$ with new measurements.
			\STATE Update the parameters $\theta_{\bm{H}}$ by solving the optimization problem (\ref{opt_prob_reduced}).
			\STATE Track the nominal trajectory with the controller $\bm{u}_{\text{ol}} =  \bm{u}_{\text{ccm}}- \bm{H}(k,\bm{x};\theta_{\bm{H}})$.
			\ENDIF
			\STATE $k=k+1$.
			\ENDWHILE
		\end{algorithmic} 
	\end{algorithm} 
\end{figure}

\section{Numerical Simulations}
In this section, two case studies are performed to validate the effectiveness of the proposed online learning scheme. The metric $\bm{M}(\bm{x};\theta_{\bm{M}})$ is formulated as $\bm{M}(\bm{x};\theta_{\bm{M}})=\underline{\alpha}\bm{I}+\bm{m}(\bm{x};\theta_{\bm{M}})^{\mathrm{T}}\bm{m}(\bm{x};\theta_{\bm{M}})$ where $\bm{m}(\bm{x};\theta_{\bm{M}})$ is structured as a two-layer neural network architecture, with each hidden layer containing 128 neurons. The feedback controller is designed by $\bm{u}(\bm x,\bm x^*,\bm u^*;\theta_{\bm{u}})=\bm{u}^*+w_2\cdot \text{tanh}(w_1\cdot(\bm x-\bm x^*))$, which satisfies $\bm{u}=\bm{u}*$ when $\bm{x}=\bm{x}^*$. And $w_1$ and $w_2$ are parameterized by two two-layer neural networks, with each hidden layer containing 128 neurons. The number of neurons in the hidden layers of $\bm{H}(t,\bm{x};\theta_{\bm{H}})$ is 128. The neural networks are trained with the Adam optimizer. For searching robust tracking controller, the training set comprises 130K sample points uniformly distributed across the data space $\mathcal{D}$. The hyber-parameters are set as $\underline{w}=0.1$, $\overline{w}=10$, and $\lambda=0.5$.

\subsection{Tethered Space Robot}
The tethered space robot (TSR) has extensive applications in space missions such as space debris removal and deep space exploration. The perturbed dynamics of TSR is given by
\begin{equation}
	\dot{\bm{x}} = \bm{f}(\bm{x})+\bm{B}(\bm{u}+\bm{h}(t,\bm{x}))
\end{equation}
where $z_1=\vartheta$ is the in-plane angle, $z_2=l-1$ represents the dimensionless tether length, $z_3=\dot{\vartheta},\ z_4=\dot{l}$, $\bm{x}=[z_1,z_2,z_3,z_4]^\mathrm{T},\ \bm{u}=[u_1,u_2]^\mathrm{T}$, with
\begin{gather} 
	\bm{f}(\bm{x})=
	\left[ \begin{array}{c}
		z_3\\
		z_4\\
		-2\frac{z_4}{z_2+1}(z_3+1)-3\sin z_1\cos z_1\\
		\left( z_2+1 \right) \left[ \left( z_3+1 \right) ^2+3\cos ^2z_1-1 \right]\\
	\end{array} \right] \notag
	\\
	\ \ \ \ \ \ \bm{B}=\begin{bmatrix}
		0 & 0 & 1 & 0 \\
		0 & 0 & 0 & 1
	\end{bmatrix}^\mathrm{T} \notag
	\\
	\ \ \ \ \ \ \bm{h}(t,\bm{x})=\left[ \begin{array}{c}
		0.3(\cos(t)+\sin(z_2)) \\
		0.2\cos(z_1)	
		\end{array}  \right] \label{dynamics_dimensionless}
\end{gather}
where $\bm{x}\in\mathcal{X}:=\{\bm{x}\in \mathbb{R}^n| -1<z_2\le0\}$, $\bm{u}\in \mathcal{U}:=\{\bm{u}\in\mathbb{R}^m|u_2\le0\}$, $n=4$, $m=2$. The objective is to deploy the TSR from an initial position to the final position while keeping avoid colliding with space debris. The reference trajectory is generated from the motion planner proposed in \cite{Jin2024d} for avoiding space debris. The size of memory buffer is set to 20, i.e., $N=20$. The parameters $\theta_{\bm{H}}$ of $\bm{H}(t,\bm{x};\theta_{\bm{H}})$ is updated for 2 epochs per control interval. Note that the external disturbances $\bm{h}(t,\bm{x})$ is unknown for pipeline of our proposed framework. 

\begin{figure}[!t]
	\centering
	\includegraphics[width=20.5pc]{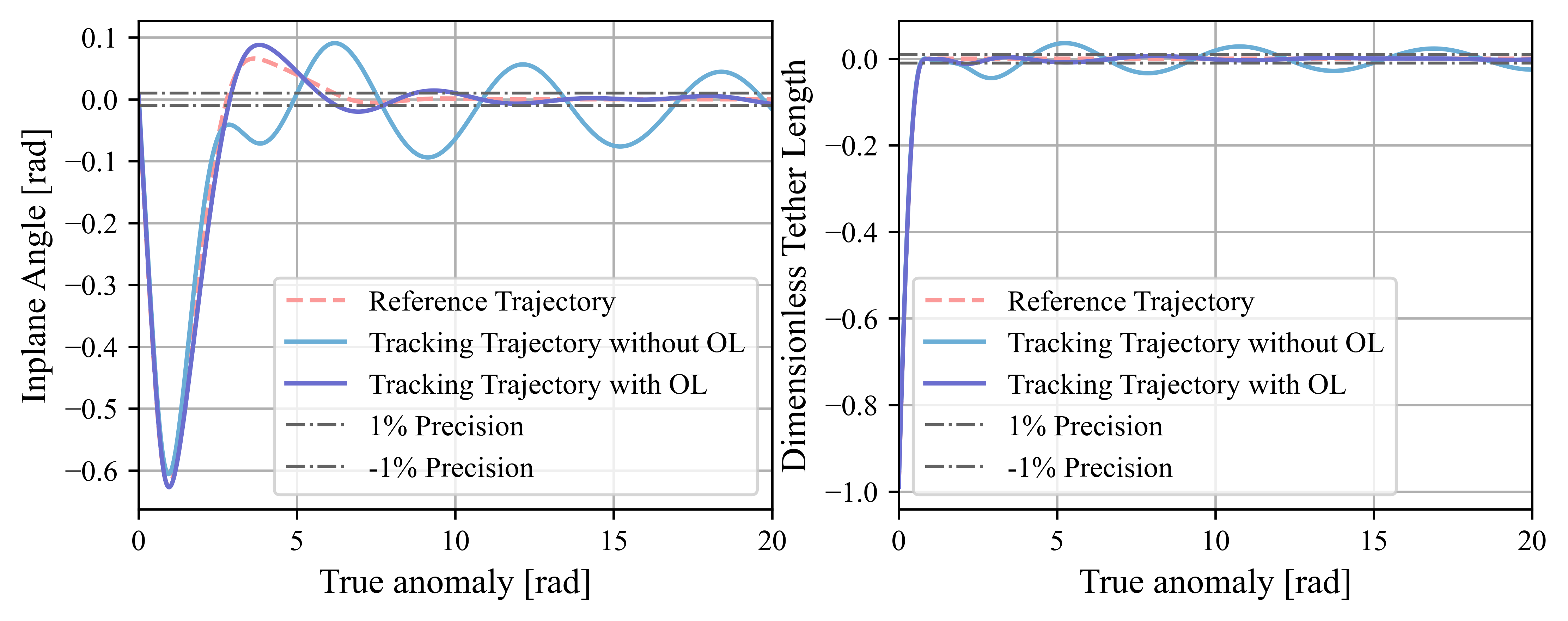} 
	\caption{The tracking performance of proposed framework for TSR. The deployment of TSR is generally done when the inplane angle and tether length reach the $1\%$ precision zone.} \label{track_TSR}
\end{figure}

The results of tracking performance are presented in Fig. \ref{track_TSR}. The tracking controller learned from neural network-driven CCM synthesis is capable of tracking reference trajectory with guaranteed performance. However, as the existence of external disturbances, the learned tracking controller $\bm{u}_{\text{ccm}}$ can not deploy TSR to the $1\%$ precision zone. The RMSE of tracking error with online learnig is $0.01051$, while the RMSE of tracking error without online learning is $0.04431$. The proposed online learning scheme improves the tracking performance by $76.3\%$, and successfully deploys the TSR to the $1\%$ precision zone. This indicates that the proposed online learning scheme can effectively capture the unknown external disturbances from online data. In addition, the computation time\footnote{The time is measured in an embedded computer with an Intel Core i5-1340p CPU.} of the proposed online learning is also investigated, which is summarized in Table \ref{computation_time}. The computation time demonstrates a nearly linear scaling with respect to the capacity of memory buffer $T$, indicating that the proposed algorithm can be efficiently solved even for large-scale problem instants. 

\begin{table}[!t]
	\renewcommand{\arraystretch}{1.3}
	\caption{Computation time of solving optimization problem (\ref{opt_prob_reduced}) with different capacity of memory buffer.}
	\centering
	\label{computation_time}
	\resizebox{8cm}{!}{
		\begin{threeparttable}
			\begin{tabular}{cccc}
				\hline
				\multicolumn{1}{c}{Capacity $T$} & \multicolumn{1}{c}{10} & \multicolumn{1}{c}{20} & \multicolumn{1}{c}{40}  \\ 
				\hline
				Time [ms] &3.8582$\pm$0.2318 & 8.5008 $\pm$ 0.6908 & 17.9657$\pm$1.5131 \\
				\hline
			\end{tabular}
		\end{threeparttable}
	}
\end{table}

\subsection{PVTOL}
The PVTOL models a planar vertical take-off and landing system for drones, which is a benchmark problem in control community. The dynamics of PVTOL is adopted from \cite{Singh2023}. The reference trajectory is randomly generated under the consideration of dynamic constraints. The size of memory buffer is set to 10, i.e., $N=10$. The parameters $\theta_{\bm{H}}$ of $\bm{H}(t,\bm{x};\theta_{\bm{H}})$ is updated for 2 epochs per control interval. The external disturbances are given by 
\begin{equation}
	\bm{h}(t,\bm{x})=\left[ \begin{array}{c}
		4(\cos(t)+\sin(p_z)+\cos(v_x)) \\
		2(\sin(p_x)+\cos(v_z))	
		\end{array}  \right]
\end{equation}

The tracking performance of position is presented in Fig. \ref{track_PVTOL}. To demonstrate the effectiveness of the proposed scheme, the start point is placed at a distance from the reference trajectory. Due to the robust properties of CCM synthesis, the tracking trajectory with controller $\bm{u}_{\text{ccm}}$ remains within the tube (\ref{RCI_set}). However, the trajectory has reached the boundary of the tube, which indicates that the tracking performance under the controller $\bm{u}_{\text{ccm}}$ is limited. The tracking trajectory with controller $\bm{u}_{\text{ol}}$ is closer to the reference trajectory, and can provide a reduced conservative motion planning for tracking as shown in Theorem \ref{theorem_rci_refined}. Moreover, with the online learning scheme, the vehicle successfully reaches the end point, whereas the vehicle under the controller $\bm{u}_{\text{ccm}}$ fails to reach the end point. The RMSE of position tracking error with online learning is $0.01772$, while while without online learning, it is $0.06448$. The proposed online learning scheme improves the tracking performance by $72.5\%$.

\begin{figure}[!t]
	\centering
	\includegraphics[width=20.5pc]{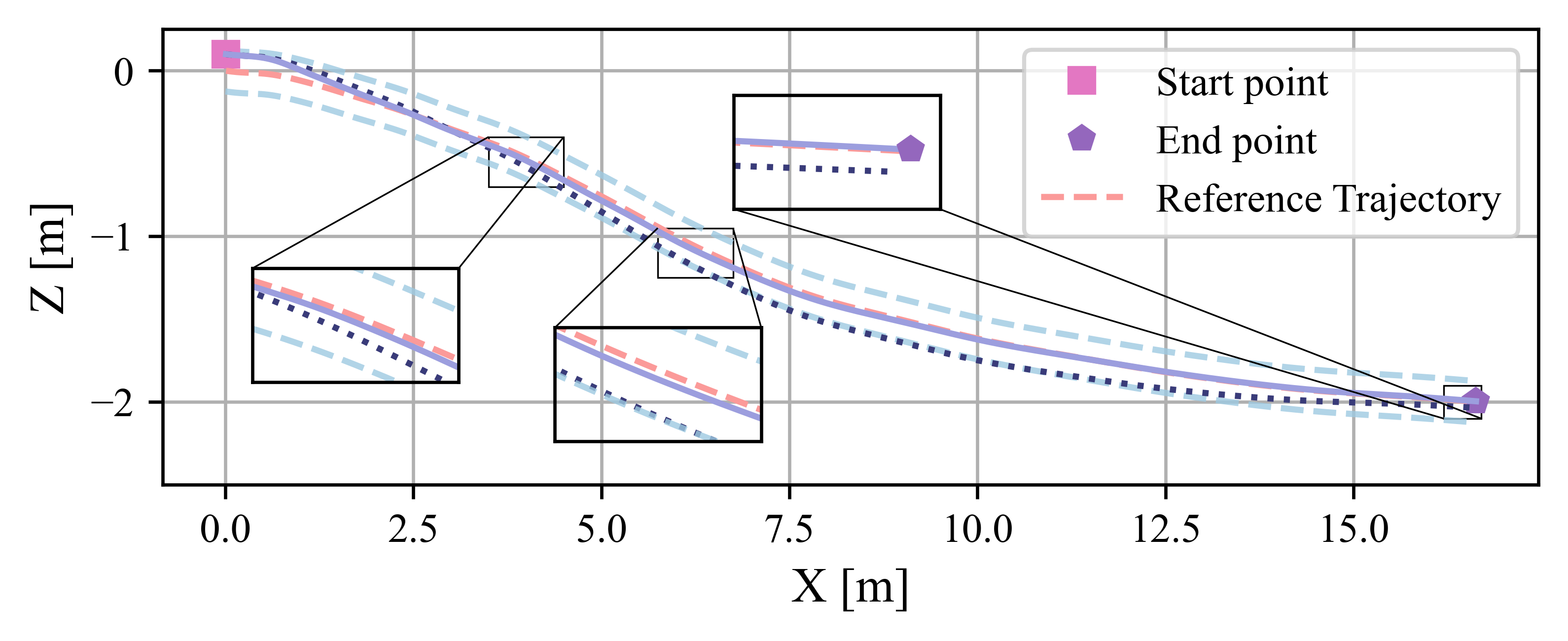} 
	\caption{The tracking performance of proposed framework for PVTOL. The black dotted line denotes the trajectory without online learning ($\bm{u}=\bm{u}_{\text{ccm}}$), the solid purple line represents the trajectory with online learning ($\bm{u}=\bm{u}_{\text{ol}}$). The light blue dashed lines represent the boundary of tube computed by (\ref{RCI_set}).} \label{track_PVTOL}
\end{figure}

In all, the proposed online learning scheme for capturing the unknown external disturbances brings two advantages over the traditional CCM synthesis. First, the tracking performance of the perturbed closed-loop system is enhanced by the proposed controller. Second, the motion planner can provide a reduced conservative motion planning for tracking.

\section{Conclusion}
This work presents an online learning scheme for enhancing the robust tracking performance of CCM synthesis. The proposed framework learns the unknown external disturbances from online historical data, enabling real-time compensation without prior knowledge. The tracking performance of the perturbed closed-loop system is improved, and the tightened state and control constraints can be refined for a reduced conservative motion planning. The proposed scheme is validated on two benchmark problems, and the results demonstrate the effectiveness of the proposed scheme. Future work will focus on the extension of the proposed scheme to the multi-agent systems and the application to the real-world scenarios.

\bibliographystyle{IEEEtran}
\bibliography{bibtex}

\end{document}